\documentclass[letterpaper, 10 pt, conference]{ieeeconf}

\IEEEoverridecommandlockouts    

\usepackage{mathptmx} 
\usepackage{times} 
\usepackage{amsmath} 
\usepackage{amssymb}  
\usepackage{amsthm}
\usepackage{amsxtra}
\usepackage{amsfonts}

\usepackage{graphicx}

\usepackage{color}

\usepackage{pict2e}

\usepackage[backend=biber,style=ieee]{biblatex}
\usepackage{hyperref}
\hypersetup{
	colorlinks=true,
	linkcolor=black,
    citecolor=black,
	filecolor=magenta,
	urlcolor=cyan,
}

\usepackage{bbm}
\DeclareSymbolFont{bbold}{U}{bbold}{m}{n}
\DeclareSymbolFontAlphabet{\mathbbm}{bbold}

\DeclareMathAlphabet{\mathcal}{OMS}{cmsy}{m}{n}

\makeatletter
\DeclareFontEncoding{LS1}{}{}
\DeclareFontSubstitution{LS1}{stix}{m}{n}
\DeclareMathAlphabet{\mathscr}{LS1}{stixscr}{m}{n}
\makeatother

\newcommand{\trs}{\prime}
\newcommand{\tr}{\mathrm{tr}}
\newcommand{\diag}{\mathrm{diag}}
\DeclareMathOperator*{\argmin}{arg\,min}

\theoremstyle{plain}
\newtheorem{theorem}{Theorem}
\newtheorem{assumption}{Assumption}
\newtheorem{corollary}{Corollary}

\newtheorem{lemma}{Lemma}
\newtheorem{definition}{Definition}
\newtheorem{example}{Example}
\newtheorem{remark}{Remark}

\title{\LARGE \bf On receding-horizon approximation in time-varying optimal control}

\author{Jintao Sun and Michael Cantoni%
\thanks{This work was supported by a Melbourne Research Scholarship and the Australian Research Council (DP210103272).}
\thanks{J. Sun is a PhD student at the Department of Electrical and Electronic Engineering, The University of Melbourne, Melbourne, Australia
        {\tt\small jintaos@student.unimelb.edu.au}}%
\thanks{M. Cantoni is with the Department of Electrical and Electronic Engineering, The University of Melbourne, Melbourne, Australia
        {\tt\small cantoni@unimelb.edu.au}}%
}

\begin{document}

\maketitle

\begin{abstract}
The closed-loop stability and infinite-horizon performance of receding-horizon approximations are studied for non-stationary linear-quadratic regulator (LQR) problems. The approach is based on a lifted reformulation of the optimal control problem, under assumed uniform controllability and observability, leading to a strict contraction property of the corresponding Riccati operator. 
Leveraging this contraction property, a stabilizing linear time-varying state-feedback approximation of the infinite-horizon optimal control policy is constructed to meet a performance-loss specification.
Its synthesis involves only finite preview of the time-varying problem data at each time step, over a sufficiently long prediction horizon. 
\end{abstract}



\section{Introduction}

Consider the following infinite-horizon linear-quadratic regulator (LQR) problem:
\begin{subequations}\label{eq:opt_problem}
\begin{align}
\min_{u} &\sum_{k=0}^{\infty} x_{k}^{\trs}Q_k x_k + u_k^{\trs}R_k u_k, \label{eq:nominal_cost} \\
\intertext{where $x_0=\xi$, and}
x_{k+1}&=A_k x_k + B_k u_k, \quad k\in\{0,1,2,\ldots\}, \label{eq:system_dynamics}
\end{align}
\end{subequations}
for given problem data $(A_k,B_k,Q_k,R_k)$ and initial state $\xi$.
The task is to determine the cost minimizing infinite-horizon control input sequence $u=(u_0,u_1,\dots)$ and corresponding state sequence $x=(x_0,x_1,\dots)$.

 Under stabilizability and detectability conditions, a linear state-feedback characterization of the optimal control policy is well-known~\cite{anderson2007optimal,bertsekas2012dynamic}. It involves the stabilizing solution of a corresponding non-stationary infinite-horizon Riccati recursion. As such, without ab initio knowledge of the problem data over all time, one can only approximate the optimal policy at each time step. 

In this paper, a lifted reformulation of~\eqref{eq:opt_problem} is employed 
to construct a receding-horizon approximation of the infinite-horizon optimal control policy. The key feature of the approximation is the use of only finite preview of the problem data over a prediction horizon at each time step. 
It is established that the receding-horizon policy is exponentially stabilizing. Further, explicit bounds are given for setting the prediction horizon length to achieve specified tolerance of infinite-horizon performance degradation. This is the main contribution. The receding-horizon approximation is ultimately a linear time-varying state-feedback controller.  

 Related time-varying work appears in~\cite{keerthi1988optimal} and~\cite{lin2021perturbation}. In~\cite{keerthi1988optimal}, the analysis of infinite-horizon performance degradation is somewhat implicit. By contrast, the explicit bounds provided here can be used for direct synthesis of a receding-horizon approximation that achieves a performance-loss specification. In~\cite{lin2021perturbation}, a Smoothed Online Convex Optimization (SOCO) sensitivity analysis is used to study the performance of a receding-horizon approximation of the optimal policy for a possibly long but finite horizon problem in terms of dynamic regret. The distinguishing feature of the developments presented below relates to the way prediction horizon length is linked to {\em infinite-horizon} performance loss. In particular, the link is established via a 
 strict contraction property of the time-varying Riccati operator in the lifted domain, as recently established in~\cite{sun2023riccati}, which builds upon a foundation result from~\cite{bougerol1993kalman}. Recent Riccati contraction based analysis of receding-horizon schemes can be found in~\cite{zhang2021regret,li2022performance} for linear time-invariant (i.e., stationary) problems. Earlier related work for stationary nonlinear problems appears in~\cite{grune2008infinite}.

 The paper is structured as follows. Notation and preliminary results are presented next. The lifted reformulation of problem~\eqref{eq:opt_problem} is developed in Section~\ref{sec:lift}. Approximation by receding-horizon control policies
 is discussed in Section~\ref{sec:rha}. Results pertaining to exponential stability under receding-horizon control are given in Section~\ref{sec:stab}. Infinite-horizon performance bounds and receding-horizon controller synthesis are then considered in Sections~\ref{sec:1step} and~\ref{sec:synth}. Concluding remarks are given in Section~\ref{sec:conc}. 

\subsection{Preliminaries}
\paragraph*{Notation}
$\mathbb{N}$ denotes the set of natural numbers, $\mathbb{N}_0 := \mathbb{N} \cup \{0\}$, and  for $i\leq j\in\mathbb{N}_0$, $[i:j]:=\{k\in\mathbb{N}_0~|~i\leq k \leq j\}$. The field of real numbers is denoted by $\mathbb{R}$, and $\mathbb{R}_{>0}:=\{\gamma\in\mathbb{R}~|~\gamma>0\}$. For $m,n\in\mathbb{N}$, the Euclidean space of real-valued $n$-vectors, with norm $|\cdot|$, and the space of real $m\times n$ matrices, are denoted by $\mathbb{R}^n$, and $\mathbb{R}^{m\times n}$, respectively. The space $\mathbb{R}^n$ is implicitly associated with $\mathbb{R}^{n \times 1}$. The $n \times n$ identity matrix is denoted by~$I_n$. The $m \times n$ matrix of zeros is denoted by~$0_{m,n}$.
The transpose of $M\in\mathbb{R}^{m\times n}$ is denoted by $M^{\trs}\in\mathbb{R}^{n\times m}$; note, $|x|^2=x^\prime x$ for $x\in\mathbb{R}^n$. When it exists, the inverse of square $M\in\mathbb{R}^{n\times n}$ is denoted by $M^{-1}\in\mathbb{R}^{n\times n}$ (i.e., $M^{-1}M=MM^{-1}=I_n$); when the relevant inverses exist, the Woodbury matrix identity $(M_1+M_2 M_3 M_4)^{-1} = M_1^{-1} - M_1^{-1}M_2(M_3^{-1}+M_4M_1^{-1}M_2)^{-1}M_4M_1^{-1}$ holds. The induced norm of $M\in\mathbb{R}^{m\times n}$ is~$\|M\|_2:=\max_{|x|=1}|Mx|$.
The respective subsets of symmetric, positive semi-definite, and positive definite matrices, are denoted by  $\mathbb{S}^n:=\{M\in\mathbb{R}^{n\times n}~|~M=M^{\trs}\}$, $\mathbb{S}_{+}^{n}:=\{M\in\mathbb{S}^n~|~(\forall x\in\mathbb{R}^n)~x^\prime M x \geq 0\}$, and 
$\mathbb{S}_{++}^{n}:=\{M\in\mathbb{S}^{n}_{+}~|~(\exists \gamma\in\mathbb{R}_{>0})(\forall x\in\mathbb{R}^n)~x^{\trs}Mx \geq \gamma |x|^2 \}$. For $M\in\mathbb{S}^n$, all eigenvalues are real; the minimum value is denoted by~$\lambda_{\min}(M)=\min_{|x|=1}x^\prime Mx$, and the maximum value by $\lambda_{\max}(M)=\max_{|x|=1}x^\prime M x$. Also, given $M_1,M_2\in\mathbb{S}^n$, the notation $M_1 \prec M_2$ (resp.~$\preceq M_2$) means $(M_2-M_1)\in\mathbb{S}_{++}^n$ (resp.~$(M_2-M_1)\in\mathbb{S}_{+}^n$.) For $M\in\mathbb{S}_{+}^{n}$, $\lambda_{\min}(M)\geq 0$ and the matrix square-root is denoted by 
$M^{1/2}\in\mathbb{S}_+^n$  (i.e., $M^{1/2}M^{1/2}=M$.) The vector space of $\mathbb{R}^n$-valued sequences is denoted by $\ell(\mathbb{R}^n)$, and $(\mathbb{R}^n)^m$ denotes the $m$-times Cartesian product $\mathbb{R}^n\times\cdots\times\mathbb{R}^n$. For $w = (w_0, w_1, \dots) \in \ell(\mathbb{R}^n)$ and $a\leq b \in \mathbb{N}_0$, the vector $w_{[a:b]}:=(w_a, \dots, w_b)\in(\mathbb{R}^n)^{(b-a+1)}$; note, $w_{[a:a]}=w_a$.

\begin{definition} \label{def:Riemannian}
Given $Y,Z\in\mathbb{S}^n_{++}$, 
the Riemannian distance is 
\begin{align}\label{eq:delta_def}
\delta(Y,Z) := \sqrt{\sum_{i=1}^{n} \log^{2} \lambda_i} ,
\end{align}
where 
$\{\lambda_1, \dots, \lambda_n\} = \mathrm{spec}\{YZ^{-1}\} \subset \mathbb{R}_{>0}$ is the spectrum (i.e., set of eigenvalues) of $YZ^{-1}\in\mathbb{R}^{n\times n}$.
\end{definition}

\begin{definition}\label{def:log}
When it exists,  $\log(M)\in\mathbb{R}^{n\times n}$ is the unique matrix logarithm for which $M=\exp(\log(M))\in\mathbb{R}^{n\times n}$, where $\exp(\cdot):= \sum_{k\in\mathbb{N}_{0}} \frac{1}{k\\!}(\cdot)^k$ denotes the matrix exponential.
\end{definition}

Proofs of the following are deferred to the Appendix.

\begin{lemma}\label{lemma:logm}
Given any non-singular $T\in\mathbb{R}^{n\times n}$, and $\{\lambda_1,\ldots,\lambda_n\}\subset\mathbb{R}$, the matrix logarithm of $M=T^{-1}\diag(\lambda_1,\dots,\lambda_n)T$
is
\begin{align*}
\log(M) = T^{-1}\diag(\log(\lambda_1),\dots,\log(\lambda_n))T.
\end{align*}
\end{lemma}

\begin{lemma}\label{lemma:log_symmetric}
If $M \in \mathbb{S}^n_{++}$, then $\log(M) \in \mathbb{S}^n$.
\end{lemma}

\begin{lemma}\label{lemma:delta_invariant}
For any $Y,Z \in \mathbb{S}^n_{++}$,
\begin{align}\label{eq:delta_equality}
\delta(Y,Z) = \|\log(Z^{-\frac{1}{2}}YZ^{-\frac{1}{2}})\|_F .
\end{align}
\end{lemma}

\begin{lemma}\label{lemma:log_norm}
If $M \in \mathbb{S}^n_{++}$ with $\lambda_{\min}(M) \geq 1$, then
\begin{align*}
\|\log(M)\|_2 = \log(\|M\|_2) .
\end{align*}
\end{lemma}

\begin{lemma}\label{lemma:u-v_bound}
For any $Y,Z \in \mathbb{S}^n_{++}$ with $Y \succeq Z$,
\begin{align*}
\|Y - Z\|_2 \leq \|Z\|_2 (\exp({\delta(Y,Z)}) - 1) ,
\end{align*}
where $\delta(\cdot,\cdot)$ is the Riemannian metric in~\eqref{eq:delta_def}.
\end{lemma}

\section{A lifted reformulation}\label{sec:lift}
A lifting approach is used to transform \eqref{eq:opt_problem} into an LQR problem that is $1$-step controllable and observable; see~\cite{sun2023riccati}. Given problem data $A_k\in\mathbb{R}^{n\times n}$, $B_{k}\in\mathbb{R}^{n\times m}$, $Q_k \in \mathbb{S}_{+}^n$, $R_k \in \mathbb{S}_{++}^m$, for $k \in \mathbb{N}_0$, define the $d$-step state transition matrix, $d\in\mathbb{N}$, by
\begin{align}\label{eq:state_transition}
\Theta_{k,d} := A_{k+d-1} A_{k+d-2} \cdots A_{k},
\end{align}
with $\Theta_{k,0} := I_n$.
The corresponding $d$-step controllability matrix is defined by
\begin{align}\label{eq:ctr_matrix}
\mathcal{C}_{k,d} :=
\begin{bmatrix}
\Theta_{k+1,d-1}B_k & \Theta_{k+2,d-2}B_{k+1} & \dots & \Theta_{k+d,0}B_{k+d-1}
\end{bmatrix} .
\end{align}
Further, the $d$-step observability matrix is defined by
\begin{align}\label{eq:obs_matrix}
\mathcal{O}_{k,d} :=
\begin{bmatrix}
C_k \Theta_{k,0} \\
C_{k+1} \Theta_{k,1} \\
\vdots \\
C_{k+d} \Theta_{k,d}
\end{bmatrix} ,
\end{align}
where $C_k := Q_k^{\frac{1}{2}}$.

\begin{assumption}\label{asm:bounded_data}
There exist $\overline{a}, \overline{b}, \overline{q}, \overline{r} \in \mathbb{R}_{>0}$ such that for all $k \in \mathbb{N}_0$, $\|A_k\|_2 \leq \overline{a}$, $\|B_k\|_2 \leq \overline{b}$, $\|Q_k\|_2 \leq \overline{q}$, and $\|R_k\|_2 \leq \overline{r}$.
\end{assumption}

\begin{assumption}\label{asm:a_invertible}
$A_k$ in \eqref{eq:system_dynamics} is non-singular for all $k \in \mathbb{N}_0$. This is a standard assumption, which holds when \eqref{eq:system_dynamics} arises from the discretization of continuous-time dynamics, for example.
\end{assumption}

\begin{assumption}\label{asm:uni_ctr_obs}
There exists $d \in \mathbb{N}$ such that for all $k \in \mathbb{N}_0$, 
$\mathcal{C}_{k,d}$ has full row rank, and 
$\mathcal{O}_{k,d}$ has full column rank.
\end{assumption}



With $d \in \mathbb{N}$ as per Assumption~\ref{asm:uni_ctr_obs}, define
\begin{subequations}  \label{eq:ABhat}
\begin{align}
\hat{A}_t &:= I_{n(d+1)} -
\begin{bmatrix}
0_{n,nd} & 0_{n,n} \\
\diag(A_{dt}, \dots, A_{d(t+1)-1}) & 0_{nd,n}
\end{bmatrix}, \label{eq:Ahat} \\
\hat{B}_t &:= \begin{bmatrix}
0_{n,md} \\ \diag(B_{dt}, \dots, B_{d(t+1)-1})
\end{bmatrix} , \label{eq:Bhat} \\
\hat{C}_t &:=
\begin{bmatrix}
\diag(C_{dt}, \dots, C_{d(t+1)-1}) & 0_{nd,n}
\end{bmatrix} ,
\end{align}
\end{subequations}
for $t \in \mathbb{N}_0$, where $C_t := Q_t^{\frac{1}{2}}$.
With
\begin{subequations}
\begin{align}
\Phi_t &:= \begin{bmatrix}
0_{n,nd} & I_n
\end{bmatrix}
\hat{A}_t^{-1}
\begin{bmatrix}
I_n \\ 0_{nd,n}
\end{bmatrix} , \label{eq:phi_compact} \\
\Gamma_t &:= \begin{bmatrix}
0_{n,nd} & I_n
\end{bmatrix}
\hat{A}_t^{-1} \hat{B}_t \label{eq:gamma_compact}, \\
\Xi_t &:= \hat{C}_t \hat{A}_t^{-1} \begin{bmatrix}
I_n \\ 0_{nd,n}
\end{bmatrix} , \label{eq:xi_compact} \\
\Delta_t &:= \hat{C}_t \hat{A}_t^{-1} \hat{B}_t, \label{eq:delta_compact}
\end{align}
\end{subequations}
the lifted problem data comprises
\begin{subequations}\label{eq:lifted_problem_data}
\begin{align}
&\tilde{Q}_t := \Xi_t^{\trs}\Xi_t - \Xi_t^{\trs}\Delta_t\tilde{R}_t^{-1}\Delta_t^{\trs}\Xi_t , \label{eq:tilde_q} \\
&\tilde{R}_t := \diag(R_{dt}, \dots, R_{dt+d-1}) + \Delta_t^{\trs}\Delta_t , \label{eq:tilde_r} \\
&\tilde{A}_t := \Phi_t - \Gamma_t \tilde{R}_t^{-1} \Delta_t^{\trs} \Xi_t , \label{eq:tilde_a} \\
&\tilde{B}_t := \Gamma_t. \label{eq:tilde_b}
\end{align}
\end{subequations}



\begin{remark} \label{rem:lifted_props}
Under Assumption~\ref{asm:a_invertible}, and with $d\in\mathbb{N}$ as per Assumption~\ref{asm:uni_ctr_obs}, $\tilde{A}_t$ is non-singular~\cite[Lem.~6]{sun2023riccati}, $\tilde{B}_t$ is full row rank~\cite[Rem.~1]{sun2023riccati}, and $\tilde{Q}_t \in \mathbb{S}^n_{++}$~\cite[Lem.~5]{sun2023riccati}, for all $t \in \mathbb{N}_0$.
\end{remark}

\begin{assumption} \label{asm:ctr_obs_uniform}
    The observability and controllability property in Assumption~\ref{asm:uni_ctr_obs} holds uniformly in the sense that $\inf_{t\in\mathbb{N}_0} \lambda_{\min}(\tilde{Q}_t) >0$, and $\inf_{t\in\mathbb{N}_0} \lambda_{\min}(\tilde{B}_t\tilde{B}_t^{\trs}) > 0$. It is also assumed that $\inf_{t \in \mathbb{N}_0} \lambda_{\min}(\tilde{A}_t \tilde{A}_t^{\trs}) > 0$.
\end{assumption}

 Given $T \in \mathbb{N}$, for $t \in \mathbb{N}_0$, $\chi \in \mathbb{R}^n$, and $w=(w_t,\dots,w_{t+T-1}) \in ((\mathbb{R}^{m})^{d})^T$, let
\begin{subequations}\label{eq:j_def}
\begin{align}
J_T(t,\chi,w) := \sum_{j=t}^{t+T-1} 
z_j^{\trs}\tilde{Q}_j z_j  
+ w_j^{\trs}\tilde{R}_jw_j
\end{align}
 where 
 \begin{align}
 z_{j+1} = \tilde{A}_j z_j + \tilde{B}_jw_j,~ j \in [t:t+T-1],
 \end{align}
\end{subequations}
with $z_t = \chi$, and  $\tilde{Q}_t, \tilde{R}_t, \tilde{A}_t, \tilde{B}_t$, as given in~\eqref{eq:lifted_problem_data}.
For 
$\tilde{u} \in \ell((\mathbb{R}^{m})^{d})$ define the infinite-horizon cost at time $t\in\mathbb{N}_0$ as
\begin{align} \label{eq:infhorizoncost}
J(t,\chi,\tilde{u}) := \lim_{T \to \infty} J_T(t,\chi,\tilde{u}_{[t:t+T-1]}),
\end{align}
and the optimal cost-to-go as
\begin{align} \label{eq:liftedform}
V_{\infty}(t,\chi) := \min_{\tilde{u} \in \ell((\mathbb{R}^{m})^{d})} J(t,\chi,\tilde{u}).
\end{align}
It is well-known (e.g., see~\cite[Section~3.3]{anderson2007optimal} and \cite{bertsekas2012dynamic}) that
\begin{align}\label{eq:opt_terminal}
V_{\infty}(t,\chi) = \chi^{\trs}P_t\chi ,
\end{align}
where $(P_t)_{t\in\mathbb{N}_0}\subset\mathbb{S}_{+}$ is the 
bounded positive semi-definite solution of the (backward) recursion
\begin{align}\label{eq:recursion}
P_t = \mathcal{R}_t(P_{t+1}),
\end{align}
where the Riccati operator is defined by
\begin{align} \label{eq:ricc_def}
\mathcal{R}_t(P) 
&:= \tilde{Q}_t + \tilde{A}_t^{\trs} (P - P\tilde{B}_t(\tilde{R}_t + \tilde{B}_t^{\trs}P\tilde{B}_t)^{-1}\tilde{B}_t^{\trs}P) \tilde{A}_t 
\end{align}
for $P\in\mathbb{S}^n$;
 n.b., the recursion does not have a boundary condition. Existence and uniqueness of a bounded positive semi-definite solution $(P_t)_{t\in\mathbb{N}_0}$ is guaranteed under Assumption~\ref{asm:uni_ctr_obs}, which implies uniform stabilizability and detectability~\cite{de1992time}. The optimal control input  
$\tilde{u}^*=\argmin_{\tilde{u}\in\ell((\mathbb{R}^{m})^d)}J(0,\xi,\tilde{u})$ 
corresponds to the 
linear 
state-feedback control policy
\begin{align} \label{eq:optpolicy}
\tilde{u}_t =\mu^*(t,\tilde{x}_t) := -(\tilde{R}_t + \tilde{B}_t^{\trs} P_{t+1} \tilde{B}_t)^{-1} \tilde{B}_t^{\trs}P_{t+1}\tilde{A}_t \tilde{x}_t
\end{align}
for the dynamics of the lifted system state given by
\begin{align} 
\label{eq:xj_evolve}
\tilde{x}_{t+1} = \tilde{A}_t\tilde{x}_t+\tilde{B}_t\tilde{u}_t,\quad t\in\mathbb{N}_0,
\end{align}
with initial condition $\tilde{x}_0=\xi$.
In particular, $\tilde{u}_t^* = \mu^*(t,\tilde{x}_t)$~\cite{anderson2007optimal,bertsekas2012dynamic}. The following result is taken from~~\cite[Lemma~4]{sun2023riccati}.
\begin{lemma}\label{lemma:equivalence}
The optimal cost associated with the original infinite-horizon problem~\eqref{eq:opt_problem} is equal to~$V_{\infty}(0,\xi)=\xi^{\trs} P_0\xi$.
\end{lemma}

\begin{remark}
As shown in the proof of~\cite[Lemma~4]{sun2023riccati},
the solution~$u^* \in \ell(\mathbb{R}^m)$ of the original problem~\eqref{eq:opt_problem} can be recovered from \eqref{eq:optpolicy} as follows:
\begin{align*}
u_{[dt:d(t+1)-1]}^*
= \mu^*(t,\tilde{x}_t) 
- \tilde{R}_t^{-1}\Delta_t^{\trs}\Xi_t\tilde{x}_t.
\end{align*}
\end{remark}

\begin{remark} \label{rem:PDricc}
Since $\tilde{R}_t,\tilde{R}_t^{-1}\in\mathbb{S}^{md}_{++}$,
\begin{align*}
D(P)&:=(P-P\tilde{B}_t(\tilde{R}_t+\tilde{B}^{\trs}P\tilde{B}_t)^{-1}\tilde{B}_t^{\trs}P) \\
&=P^{1/2}(I +  P^{1/2}\tilde{B}_t\tilde{R}_t^{-1}\tilde{B}_t^{\trs}P^{1/2})^{-1}P^{1/2}\in\mathbb{S}^n_{+}
\end{align*}
for all $P\in\mathbb{S}_+$; the equality holds by the Woodbury matrix identity.
So $\mathcal{R}_t(P) = \tilde{Q}_t + \tilde{A}_t^\prime D(P) \tilde{A}_t \in \mathbb{S}^n_{++}$, because $\tilde{Q}_t\in \mathbb{S}^n_{++}$ as noted in Remark~\ref{rem:lifted_props}. Therefore, the unique positive semi-definite solution $(P_t)_{t\in\mathbb{N}}$ of \eqref{eq:recursion} is positive definite, with $\lambda_{\min}(P_t) \geq \lambda_{\min}(\tilde{Q}_t)>0$.
\end{remark}

Within the time-varying context of this work, it is important to note that implementation of 
\eqref{eq:optpolicy}
requires knowledge of the lifted problem data $(\tilde{A}_j,\tilde{B}_j, \tilde{Q}_j, \tilde{R}_j)$ 
for all $j \geq t$, as needed to determine $P_{t+1}$ according to \eqref{eq:recursion}, and thus, 
$u_t^*$ at time $t\in\mathbb{N}_0$. To overcome this impediment, a so-called receding-horizon approximation can be employed.
The impact of such an approximation on stability and performance is investigated in the subsequent developments.

\section{Receding-horizon approximation}\label{sec:rha}

A receding-horizon scheme is presented here to approximate the optimal policy~\eqref{eq:optpolicy} for the lifted reformulation of the original infinite-horizon LQR problem~\eqref{eq:opt_problem}; i.e., for~\eqref{eq:liftedform} with $t=0$. 
With $T \in \mathbb{N}$, and $d\in\mathbb{N}$ as per Assumption~\ref{asm:uni_ctr_obs}, define 
$L_T:\mathbb{N}_0\times \mathbb{R}^n\times ((\mathbb{R}^{m})^{d})^T \times \mathbb{S}_{++}^n \rightarrow \mathbb{R}_{\geq 0}$
by\begin{align}\label{eq:j_tilde_def}
L_T(t,\chi,w,X) :=  
J_T(t,\chi,w) + z_{t+T}^{\trs}X z_{t+T}
\end{align}
with $J_T$ and $z_{t+T}$ as per \eqref{eq:j_def} for the given $w=(w_t,\dots,w_{t+T-1})$.
Given bounded terminal penalty matrix sequence $(X_{t+T})_{t\in\mathbb{N}_0} \subset \mathbb{S}^n_{++}$, the $T$-step receding-horizon scheme is defined by the state-feedback control policy
\begin{subequations}\label{eq:MPC_policy}
\begin{align}
    \tilde{u}_t = \mu^{RH}_0(t,\tilde{x}_t),
\end{align}
where
\begin{align}(\mu_0^{RH}(t,\tilde{x}_t), \dots, \mu_{T-1}^{RH}(t,\tilde{x}_t)) = \argmin_{w \in ((\mathbb{R}^{m})^{d})^T} L_T(t,\tilde{x}_t,w,X_{t+T}).
    \end{align}
\end{subequations}
%

\begin{remark}
The receding-horizon policy~\eqref{eq:MPC_policy} corresponds to a sampled-data policy in the domain of the original problem:
\begin{align}\label{eq:rh_unlifted}
u_{[td:(t+1)d-1]} = 
\mu_0^{RH}(t,x_{td})
-\tilde{R}_t^{-1}\Delta_t^{\trs}\Xi_t x_{td},
\end{align}
where 
$x_k$ evolves according to \eqref{eq:system_dynamics} from $x_0=\xi$.
\end{remark}


The quality of this receding-horizon policy, as an approximation of the optimal policy~\eqref{eq:optpolicy}, is investigated subsequently
in terms of both closed-loop stability and infinite-horizon performance degradation.
%
%
%
%
First, it is noted
that 
the receding-horizon policy~\eqref{eq:MPC_policy} recovers the least infinite-horizon cost with particular terminal penalty matrices; however, this requires ab initio knowledge of the problem data over the infinite horiozon.
Given $T \in \mathbb{N}$, define 
$W_T : \mathbb{N}_0 \times \mathbb{R}^n \times \mathbb{S}^n_{++} \to \mathbb{R}_{\geq 0}$  by
\begin{align}\label{eq:v_tilde_def}
W_T(t,\chi,X) := \min_{w \in ((\mathbb{R}^{m})^{d})^T} \ 
L_T(t,\chi,w,X),
\end{align}
where $L_T$ is defined in \eqref{eq:j_tilde_def}.

\begin{lemma}\label{lemma:p_opt}
For all $t \in \mathbb{N}_0$, and $\chi \in \mathbb{R}^n$,
\begin{align} \label{eq:WTquad}
W_T(t,\chi,X) = \chi^{\trs} (\mathcal{R}_{t}\circ \cdots \circ \mathcal{R}_{t+T-1}(X))\chi, 
\end{align}
with $\mathcal{R}_{\bullet}(\cdot)$ as per \eqref{eq:ricc_def}.
In particular,
\begin{align} \label{eq:WTP}
W_T(t,\chi,P_{t+T}) =  \chi^{\trs} P_{t} \chi = V_{\infty}(t,\chi), 
\end{align}
where $(P_{t})_{t\in\mathbb{N}_0}$ is the bounded positive definite solution of~\eqref{eq:recursion}.
\end{lemma}

\begin{proof}
It is well known, e.g., see~\cite[Section~2.4]{anderson2007optimal} and \cite{bertsekas2012dynamic}, that the value function of the finite horizon LQ optimal control problem \eqref{eq:v_tilde_def} takes the quadratic form given in \eqref{eq:WTquad}. As such, in view of~\eqref{eq:opt_terminal} and ~\eqref{eq:recursion}, $V_\infty(t,\chi)=\chi^{\trs} P_t \chi = \chi^{\trs}
(\tilde{\mathcal{R}}_{t}\circ \cdots \circ \tilde{\mathcal{R}}_{t+T-1}(P_{t+T}))\chi^{\trs} = W_T(t,\chi,P_{t+T})$.
\end{proof}


For given initial state $\xi \in \mathbb{R}^n$, and control input $\tilde{u}\in\ell((\mathbb{R}^m)^d)$, the performance loss 
with respect to the optimal control input $\tilde{u}^*\in\ell((\mathbb{R}^{m})^{d})$ as per~\eqref{eq:optpolicy},
is defined by
\begin{align}
\label{eq:regret_def}
\beta(\xi) := J(0,\xi,\tilde{u}) - J(0,\xi,\tilde{u}^*),
\end{align}
where the performance index $J$ is given in~\eqref{eq:infhorizoncost}.
Note that $\beta(\xi)=J(0,\xi,\tilde{u}) - \xi^{\trs}P_0\xi$.
It quantifies the {\em infinite-horizon} performance degradation.
The aim 
here is to obtain an upper bound on 
~\eqref{eq:regret_def} for the control input generated by~\eqref{eq:MPC_policy}, and a method for selecting the prediction horizon $T \in \mathbb{N}$ and penalty matrix sequence
to achieve specified performance-loss tolerance.

\section{Closed-loop stability}
\label{sec:stab}
For suitable sequences of terminal penalty matrices,
the receding-horizon policy~\eqref{eq:MPC_policy} is exponentially stabilizing.

\begin{theorem}\label{theorem:mpc_stable}
 Given $T \in \mathbb{N}$, let $(X_{t+T})_{t\in\mathbb{N}_0} \subset \mathbb{S}^n_{++}$ be any bounded terminal penalty matrix sequence such that
\begin{align}\label{eq:sufficient_stability}
X_{t+T} \succeq \mathcal{R}_{t+T}(X_{t+T+1})
\end{align}
for all $t \in \mathbb{N}_0$, with $\mathcal{R}_{t+T}(\cdot)$ as per \eqref{eq:ricc_def}. Then, the origin is exponentially stable for \eqref{eq:xj_evolve} under the corresponding receding-horizon state-feedback control policy~\eqref{eq:MPC_policy}. 
\end{theorem}

\begin{proof}
With reference to \eqref{eq:v_tilde_def}, note that
$W_T(t,\tilde{x}_t,X_{t+T}) = \tilde{x}_{t}^{\trs} (\mathcal{R}_{t} \circ \cdots \circ \mathcal{R}_{t+T-1}(X_{t+T})) \tilde{x}_t$ by Lemma~\ref{lemma:p_opt}. Thus,  in view of Remark~\ref{rem:PDricc}, and the hypothesis $X_{t+T}\in\mathbb{S}_{++}$, $W_T$ is positive and radially unbounded as a function of $\tilde{x}_t$. Further, it can be shown that $t \mapsto W_T(t,\tilde{x}_t,X_{t+T})$ is strictly decreasing for the evolution of $\tilde{x}_t$ according to \eqref{eq:xj_evolve} with 
$\tilde{u}_t$ as per~\eqref{eq:MPC_policy}.
For $v \in (\mathbb{R}^{m})^{d}$, 
define 
\begin{align*}
w_{t+1}(v) := 
(\tilde{u}_{t+1},\dots,\tilde{u}_{t+T-1}, v) \in ((\mathbb{R}^{m})^{d})^T.
\end{align*}
Then, 
\begin{align}
&L_T(t+1,\tilde{x}_{t+1},w_{t+1}(v),X_{t+T+1}) \nonumber\\
&= W_T(t,\tilde{x}_t,X_{t+T}) - \tilde{x}_t^{\trs}\tilde{Q}_t\tilde{x}_t - \tilde{u}_t^{\trs}\tilde{R}_t\tilde{u}_t + G_{t+T}(v) , \label{eq:l_t+1}
\end{align}
where
\begin{align*}
& G_{t+T}(v) \\
& := - \tilde{x}_{t+T}^{\trs}X_{t+T}\tilde{x}_{t+T} + \tilde{x}_{t+T}^{\trs}\tilde{Q}_{t+T}\tilde{x}_{t+T} + v^{\trs}\tilde{R}_{t+T}v \\
&\quad + (\tilde{A}_{t+T}\tilde{x}_{t+T}+\tilde{B}_{t+T}v)^{\trs} X_{t+T+1} (\tilde{A}_{t+T}\tilde{x}_{t+T}+\tilde{B}_{t+T}v) .
\end{align*}
By `completing-the-square', 
\begin{align*}
G_{t+T}(v)
&= (K_{t+T}\tilde{x}_{t+T} + v)^{\trs} M_{t+T} (K_{t+T}\tilde{x}_{t+T} + v) \\
& \quad + \tilde{x}_{t+T}^{\trs} (\tilde{\mathcal{R}}_{t+T}(X_{t+T+1}) - X_{t+T})\tilde{x}_{t+T},
\end{align*}
where
\begin{align*}
&K_{t+T} := M_{t+T}^{-1}\tilde{B}_{t+T}^{\trs}X_{t+T+1}\tilde{A}_{t+T}, \\
&M_{t+T} := \tilde{R}_{t+T} + \tilde{B}_{t+T}^{\trs}X_{t+T+1} \tilde{B}_{t+T}.
\end{align*}
Now, by definition,
\begin{align*}
W_T(t+1,\tilde{x}_{t+1},X_{t+T+1}) \leq L_T(t+1,\tilde{x}_{t+1},w_{t+1}(v),X_{t+T+1}),
\end{align*}
and so it follows from~\eqref{eq:l_t+1} that for all $v \in (\mathbb{R}^{m})^{d}$,
\begin{align*}
&W_T(t+1,\tilde{x}_{t+1},X_{t+T+1}) - W_T(t,\tilde{x}_t,X_{t+T}) \nonumber\\
&\quad \leq L_T(t+1,\tilde{x}_{t+1},w_{t+1}(v),X_{t+T+1}) - W_T(t,\tilde{x}_t,X_{t+T}) \nonumber\\
&\quad = - \tilde{x}_t^{\trs}\tilde{Q}_t\tilde{x}_t - \tilde{u}_t^{\trs}\tilde{R}_t\tilde{u}_t + G_{t+T}(v). 
\end{align*}
Further, since $\min_{v} G_{t+T}(v) = \tilde{x}_{t+T}^{\trs}(\tilde{\mathcal{R}}_{t+T}(X_{t+T+1}) - X_{t+T})\tilde{x}_{t+T} \leq 0$ by \eqref{eq:sufficient_stability}, 
\begin{align}\label{eq:w_decreasing}
&W_T(t+1,\tilde{x}_{t+1},X_{t+T+1}) - W_T(t,\tilde{x}_t,X_{t+T})
\nonumber \\
&\qquad\qquad\qquad 
\leq  - \tilde{x}_t^{\trs} \tilde{Q}_t \tilde{x}_t - \tilde{u}_t^{\trs}\tilde{R}_t\tilde{u}_t.
\end{align}
This bound is strictly negative for non-zero $\tilde{x}_t$ since $\inf_{t\in\mathbb{N}_0}\lambda_{\min} (\tilde{Q}_t)>0$ by Assumption~\ref{asm:ctr_obs_uniform}, and $\tilde{R}_t\succ 0$.
 Therefore, $W_T$ is a Lyapunov function for the closed-loop dynamics, and the claimed exponential stability property follows from~\cite[Theorem~5.7]{bof2018lyapunov}.
\end{proof}

%


Next, Theorem~\ref{theorem:mpc_stable} is specialized to a one-step policy, which facilitates the presentation of an explicit exponential bound on the closed-loop state trajectory.
By the dynamic programming principle~\cite{bertsekas2012dynamic},
for given $T\in\mathbb{N}$, and bounded terminal penalty matrix sequence $(X_{t+T})_{t\in\mathbb{N}_0}\subset\mathbb{S}_{++}^n$, the $T$-step policy~\eqref{eq:MPC_policy} is equivalent to the $1$-step policy 
\begin{align}
    \tilde{u}_t &= {\textstyle \argmin_{w_t \in (\mathbb{R}^m)^d}} 
    ~L_1(t,\tilde{x}_t,w_t,\hat{X}_{t+1}) \nonumber \\
    &= -(\tilde{R}_t+\tilde{B}_t^{\trs}\hat{X}_{t+1}\tilde{B}_t)^{-1}\tilde{B}_t^{\trs}\hat{X}_{t+1}\tilde{A}_t \tilde{x}_t,  \label{eq:onestepmpc}
\end{align}
 where 
$\hat{X}_{t+1}:=\mathcal{R}_{t+1}\circ\cdots\circ\mathcal{R}_{t+T-1}(X_{t+T})$ is the terminal penalty matrix for each $t\in\mathbb{N}_0$.
Note, in particular, that 
$$L_1(t,\tilde{x}_t,w_t,\hat{X}_{t+1}) = (K_t\tilde{x}_t + w_t)^{\trs}M_t(K_t\tilde{x}+w_t) + \tilde{x}_t^{\trs} \mathcal{R}_t(\hat{X}_{t+1})\tilde{x}_t,$$ where $K_t:=M_t^{-1}\tilde{B}_t^{\trs}\hat{X}_{t+1}\tilde{A}_t$ and $M_t=\tilde{R}_t+\tilde{B}_t^{\trs}\hat{X}_{t+1}\tilde{B}_t$.
Also note that, with 
Assumptions~\ref{asm:bounded_data} and~\ref{asm:ctr_obs_uniform}, 
\begin{align*}
\sup_{t\in\mathbb{N}_0}\|\tilde{Q}_t + \tilde{A}_t^{\trs}(\tilde{B}_t \tilde{B}_t^{\trs})^{-1} \tilde{B}_t^{\trs} \tilde{R}_t \tilde{B}_t^{\trs}(\tilde{B}_t \tilde{B}_t^{\trs})^{-1} \tilde{A}_t\|_2 < +\infty,
\end{align*}
which implies $\sup_{t\in\mathbb{N}_0} \|\mathcal{R}_t(X)\|_2<+\infty$ for any  $X \in \mathbb{S}^n_{+}$, 
since $
L_1(t,\chi,-K_t\chi,X) \leq L_1(t,\chi,-\tilde{B}_t^\prime(\tilde{B}_t\tilde{B}_t^\prime)^{-1}\tilde{A}_t\chi,X)$ for all $\chi\in\mathbb{R}^n$, and thus,
\begin{align}
 \mathcal{R}_t(X) \preceq \tilde{Q}_t + \tilde{A}_t^{\trs}(\tilde{B}_t \tilde{B}_t^{\trs})^{-1} \tilde{B}_t \tilde{R}_t \tilde{B}_t^{\trs}(\tilde{B}_t \tilde{B}_t^{\trs})^{-1} \tilde{A}_t.
\label{eq:RiccBounded}
\end{align}
Further, $\inf_{t\in\mathbb{N}_0}\lambda_{\min}(\mathcal{R}_t(X)) > 0$ in view of Remark~\ref{rem:PDricc} and Assumption~\ref{asm:ctr_obs_uniform}.

\begin{theorem}\label{theorem:mpc_1step_stable}
Consider the state-feedback control policy~\eqref{eq:onestepmpc}
for given 
bounded 
sequence 
$(\hat{X}_{t+1})_{t\in\mathbb{N}_0} \subset \mathbb{S}^n_{++}$. 
Suppose 
$\hat{X}_t \succeq \mathcal{R}_t(\hat{X}_{t+1})$ for all $t \in \mathbb{N}$, with $\mathcal{R}_{t}(\cdot)$ as per \eqref{eq:ricc_def}.
Then, the evolution of \eqref{eq:xj_evolve} with $\tilde{u}_t$ as per 
~\eqref{eq:onestepmpc} satisfies 
\begin{align}\label{eq:cl_exp_stable}
|\tilde{x}_t|^2 \leq \frac{\overline{\omega}}{\underline{\omega}} \Big(1-\underline{\lambda}\big/\overline{\omega}\Big)^t |\tilde{x}_0|^2
\end{align}
for all $t\in\mathbb{N}_0$, with 
\begin{align}
&
\overline{\omega} := \sup_{t \in \mathbb{N}_0} \|\mathcal{R}_{t}(\hat{X}_{t+1})\|_2 \!<\! +\infty, \label{eq:w_ub}
\\
&
\underline{\omega} := \inf_{t \in \mathbb{N}_0} \lambda_{\min} (\mathcal{R}_{t}(\hat{X}_{t+1}))\!>\! 0, 
\label{eq:w_lb}\\
\label{eq:gamma_def}
&
\underline{\lambda} := \inf_{t \in \mathbb{N}_0} 
\lambda_{\min} \Big(\tilde{Q}_t + K_t^{\trs} \tilde{R}_t K_t \Big)\leq \overline{\omega},
%
\end{align}
and $K_t:=(\tilde{R}_t+\tilde{B}_t^{\trs}\hat{X}_{t+1}\tilde{B}_t)^{-1}\tilde{B}_t^{\trs}\hat{X}_{t+1}\tilde{A}_t$.
Further, $\underline{\lambda} >0$.
\end{theorem}

\begin{proof}
In view of the hypothesis
$\hat{X}_t \succeq \mathcal{R}_t(\hat{X}_{t+1})$, \eqref{eq:w_decreasing}, and~\eqref{eq:onestepmpc},
\begin{align}\label{eq:w_1_decreasing}
& W_1(t+1,\tilde{x}_{t+1},\hat{X}_{t+2}) - W_1(t,\tilde{x}_t,\hat{X}_{t+1}) \nonumber\\
&\qquad\qquad \leq -\tilde{x}_t^{\trs}\tilde{Q}_t\tilde{x}_t - \tilde{u}_t^{\trs}\tilde{R}_t\tilde{u}_t 
\leq -\underline{\lambda} |\tilde{x}_t|^2,
\end{align}
with $\underline{\lambda}$ as per \eqref{eq:gamma_def}. Note that $\underline{\lambda} \geq \inf_{t\in\mathbb{N}_0}\lambda_{\min}(\tilde{Q}_t)>0$ by Assumption~\ref{asm:ctr_obs_uniform}, and that $\underline{\lambda} \leq \overline{\omega}$, because \begin{align*}
\tilde{x}_t^{\trs}\tilde{Q}_t\tilde{x}_t + \tilde{u}_t^{\trs}\tilde{R}_t\tilde{u}_t &= \tilde{x}_t^{\trs}\mathcal{R}_t(\hat{X}_{t+1}) \tilde{x}_t -  
\tilde{x}_{t+1}^{\trs}\hat{X}_{t+1}\tilde{x}_{t+1} \\
&\leq \tilde{x}_t^{\trs}\mathcal{R}_t(\hat{X}_{t+1}) \tilde{x}_t  
\end{align*}
for $t \in \mathbb{N}_0$. Now,
 $W_1(t,\tilde{x}_t,\hat{X}_{t+1}) = \tilde{x}_t^{\trs} \mathcal{R}_t(\hat{X}_{t+1}) \tilde{x}_t$ by Lemma~\ref{lemma:p_opt}, and as such,
\begin{align}\label{eq:v_tilde_exp_2}
\underline{\omega} |\tilde{x}_t|^2 \leq W_1(t,\tilde{x}_t,\hat{X}_{t+1}) \leq \overline{\omega} |\tilde{x}_t|^2.
\end{align}
Combining~\eqref{eq:w_1_decreasing} and~\eqref{eq:v_tilde_exp_2} yields 
\begin{align*}
W_1(t+1,\tilde{x}_{t+1},\hat{X}_{t+2}) \leq (1-\underline{\lambda}\big/\overline{\omega}) W_1(t,\tilde{x}_t,\hat{X}_{t+1}) ,
\end{align*}
for all $t \in \mathbb{N}_0$.
Therefore,
\begin{align}\label{eq:vtilde_exp}
W_1(t,\tilde{x}_t,\hat{X}_{t+1}) \leq (1-\underline{\lambda}\big/\overline{\omega})^t W_1(0,\tilde{x}_0,\hat{X}_1) .
\end{align}
Finally, using \eqref{eq:v_tilde_exp_2} 
in~\eqref{eq:vtilde_exp}, it follows that
\begin{align*}
|\tilde{x}_t|^2 \leq \frac{1}{\underline{\omega}}W_1(t,\tilde{x}_t,\hat{X}_{t+1}) 
&\leq \frac{1}{\underline{\omega}} (1-\underline{\lambda}\big/\overline{\omega})^t W_1(0,\tilde{x}_0,\hat{X}_1) \\
&\leq \frac{\overline{\omega}}{\underline{\omega}} (1-\underline{\lambda}\big/\overline{\omega})^t |\tilde{x}_0|^2,
\end{align*}
as claimed.
\end{proof}

\section{Bounded performance loss with {$T=1$}}\label{sec:1step}

As elaborated below, given any bounded $1$-step terminal penalty matrix sequence $(\hat{X}_{t+1})_{t\in\mathbb{N}_0}$ that ($i$) bounds the positive definite solution of~\eqref{eq:recursion}, 
and ($ii$) satisfies $\hat{X}_t \succeq \mathcal{R}_t(\hat{X}_{t+1})$, it is possible to bound the infinite-horizon performance loss 
associated with the corresponding state-feedback control policy~\eqref{eq:onestepmpc}, relative to the optimal policy~\eqref{eq:optpolicy}. The synthesis of such a sequence is considered in Section~\ref{sec:synth}.

\begin{lemma}\label{lemma:ric_monotone}
Given $X, P \in \mathbb{S}^n_{++}$, if $X \succeq P$,
then
$\mathcal{R}_t(X) \succeq \mathcal{R}_t(P)$
for all $t \in \mathbb{N}_0$, with $\mathcal{R}_t(\cdot)$ as per~\eqref{eq:ricc_def}.
\end{lemma}
\begin{proof}
This result is taken from~\cite[Lemma~10.1]{bitmead1991riccati}.
\end{proof}


\begin{lemma}\label{lemma:tilde_v_lipschitz}
Given $X, P \in \mathbb{S}^n_{++}$, if $X \succeq P$, then
\begin{align*}
W_1(t,\chi,X) - W_1(t,\chi,P) \leq \|P\|_2 |\chi|^2 (\exp(\delta(X,P)) - 1)
\end{align*}
for all $t \in \mathbb{N}_0$, and $\chi \in \mathbb{R}^n$, with $\delta(\cdot,\cdot)$ as per~\eqref{eq:delta_def}, and $W_1(\cdot,\cdot,\cdot)$ as per~\eqref{eq:v_tilde_def}.
\end{lemma}

\begin{proof}
By Lemma~\ref{lemma:p_opt},
$W_1(t,\chi,X) = \chi^{\trs} \tilde{\mathcal{R}}_t(X) \chi$, and $
W_1(t,\chi,P) = \chi^{\trs} \tilde{\mathcal{R}}_t(P) \chi$, whereby
\begin{align}
W_1(t,\chi,X) - W_1(t,\chi,P) &= \chi^{\trs} \tilde{\mathcal{R}}_t(X)\chi - \chi^{\trs}\tilde{\mathcal{R}}_t(P)\chi \nonumber\\
&\leq |\chi|^2 \|\tilde{\mathcal{R}}_t(X) - \tilde{\mathcal{R}}_t(P)\|_2. \label{eq:vp1vp2}
\end{align}
Given $X\succeq P$ by hypothesis, 
$\hat{X}:=\tilde{\mathcal{R}}_t(X) \succeq \tilde{\mathcal{R}}_t(P)=:\hat{P}$ by Lemma~\ref{lemma:ric_monotone}, 
and thus, 
\begin{align}\label{eq:rp1rp2}
\|\hat{X} - \hat{P}\|_2 \leq \|P\|_2 (\exp(\delta(\hat{X}, \hat{P})) - 1) 
\end{align}
by Lemma~\ref{lemma:u-v_bound}.
Since $\tilde{\mathcal{R}}_t(\cdot)$ is a contraction~\cite[Theorem~1]{sun2023riccati}, in the sense that
\begin{align*}
\delta(\hat{X},\hat{P})=\delta(\tilde{\mathcal{R}}_t(X), \tilde{\mathcal{R}}_t(P)) \leq \delta(X,P),
\end{align*}
it follows from~\eqref{eq:rp1rp2} that $\|\hat{X}-\hat{P}\|_2=\|\tilde{\mathcal{R}}_t(X)-\tilde{\mathcal{R}}_t(X)\|_2 \leq \|P\|_2(\exp(\delta(X,P))-1)$, which in combination with~\eqref{eq:vp1vp2} yields the result.
\end{proof}

\begin{theorem}\label{theorem:regret_bound}
Consider the state-feedback control policy~\eqref{eq:onestepmpc}
for given bounded
sequence 
$(\hat{X}_{t+1})_{t\in\mathbb{N}_0} \subset \mathbb{S}^n_{++}$. Suppose $$\hat{X}_t \succeq \mathcal{R}_t(\hat{X}_{t+1}),\quad \hat{X}_t \succeq P_t,\quad \text{and}\quad \delta(\hat{X}_t,P_t) \leq \eta,$$ for all $t\in\mathbb{N}$, with $\eta\in\mathbb{R}_{>0}$, $\delta(\cdot,\cdot)$ as per \eqref{eq:delta_def},
$(P_t)_{t\in\mathbb{N}_0}$ as the bounded positive definite 
solution of~\eqref{eq:recursion}, 
and $\mathcal{R}_t(\cdot)$ as per \eqref{eq:ricc_def}. Then, for initial state $\xi\in\mathbb{R}^n$, the 
performance loss as defined in~\eqref{eq:regret_def} satisfies
\begin{align*}
\beta(\xi) \leq \frac{\overline{\lambda}}{\underline{\lambda}}
\frac{\overline{\omega}}{\underline{\omega}} \overline{\omega}
(\exp(\eta)-1) |\xi|^2,
\end{align*}
with
\begin{align}\label{eq:p_ub}
\overline{\lambda} := \sup_{t \in \mathbb{N}_0} \|P_t\|_2 < +\infty,
\end{align}
and $\overline{\omega}, \underline{\omega}, \underline{\lambda} \in \mathbb{R}_{>0}$ as per \eqref{eq:w_ub}, 
\eqref{eq:w_lb}, and~\eqref{eq:gamma_def}.
\end{theorem}

\begin{proof}
Let $\tilde{x}_t$ evolve from $\tilde{x}_0=\xi$ according to \eqref{eq:xj_evolve} with the input $\tilde{u}_t$
as per 
\eqref{eq:onestepmpc}.
Further, let 
$l(t,\tilde{x}_t) := \tilde{x}_t^{\trs} \tilde{Q}_t \tilde{x}_t + \tilde{u}_t^{\trs} \tilde{R}_t \tilde{u}_t$,
noting that 
\begin{align} \label{eq:stage_lt}
l(t,\tilde{x}_t) = W_1(t,\tilde{x}_t,\hat{X}_{t+1}) - \tilde{x}_{t+1}^{\trs} \hat{X}_{t+1} \tilde{x}_{t+1},
\end{align}
where $W_1$ is given in \eqref{eq:v_tilde_def}. Since $W_1(t,\tilde{x}_t,P_{t+1}) = \tilde{x}_t^{\trs}P_t\tilde{x}_t$ by Lemma~\ref{lemma:p_opt}, and since $\hat{X}_{t+1}\succeq P_{t+1}$ by hypothesis, 
\begin{align*}
l(t,\tilde{x}_t) &\leq W_1(t,\tilde{x}_t,\hat{X}_{t+1}) - 
W_1(t,\tilde{x}_t,P_{t+1}) \nonumber \\
&\qquad\qquad\qquad\qquad + \tilde{x}_t^{\trs}P_t\tilde{x}_t 
-\tilde{x}_{t+1}^{\trs} P_{t+1} \tilde{x}_{t+1} \nonumber\\
& \leq 
\|P_{t+1}\|_2 |\tilde{x}_t|^2 (\exp(\delta(\hat{X}_{t+1},P_{t+1})) - 1 ) \nonumber \\
&\qquad\qquad\qquad\qquad + \tilde{x}_t^{\trs}P_t\tilde{x}_t 
-\tilde{x}_{t+1}^{\trs} P_{t+1} \tilde{x}_{t+1} 
\nonumber
\\
&\leq \overline{\lambda} |\tilde{x}_t|^2 (\exp(\eta) - 1 ) + \tilde{x}_t^{\trs}P_t\tilde{x}_t 
-\tilde{x}_{t+1}^{\trs} P_{t+1} \tilde{x}_{t+1}, 
\end{align*}
where $\overline{\lambda}<+\infty$ as $(P_t)_{t\in\mathbb{N}_0}$ is bounded, the second inequality holds by Lemma~\ref{lemma:tilde_v_lipschitz}, and the last follows from \eqref{eq:p_ub} and the hypothesis $\delta(\hat{X}_{t+T},P_{t+T})\leq \eta$.
As such, 
\begin{align*}
\sum_{t=0}^{N} l(t,\tilde{x}_t) \!
\leq \! \overline{\lambda}  (\exp(\eta) - 1 ) \sum_{t=0}^N|\tilde{x}_t|^2 + \xi^{\trs}P_0\xi 
-\tilde{x}_{N+1}^{\trs} P_{N+1} \tilde{x}_{N+1},
\end{align*}
for all $N\in\mathbb{N}$. Since $\tilde{x}_{N+1}\rightarrow 0$ as $N\rightarrow\infty$ by Theorem~\ref{theorem:mpc_1step_stable}, it follows from \eqref{eq:cl_exp_stable} that 
\begin{align*}
J(0,\xi,\tilde{u})
&=\lim_{N\rightarrow\infty} \sum_{t=0}^N l_t(\tilde{x}_t)\\
&\leq \overline{\lambda}  (\exp(\eta) - 1 ) \lim_{N\rightarrow\infty}\sum_{t=0}^N|\tilde{x}_t|^2 + \xi^{\trs}P_0\xi \\
&\leq \overline{\lambda}  (\exp(\eta) - 1 ) \frac{\overline{\omega}}{\underline{\omega}}\frac{\overline{\omega}}{\underline{\lambda}}|\xi|^2 + \xi^{\trs}P_0\xi.
\end{align*}
Therefore, $$\beta(\xi)= J(0,\xi,\tilde{u}) - \xi^{\trs} P_0 \xi \leq   (\exp(\eta) - 1 ) \frac{\overline{\lambda}}{\underline{\lambda}} \frac{\overline{\omega}}{\underline{\omega}} \overline{\omega} |\xi|^2,$$ as claimed.
\end{proof}

\section{Receding-horizon policy synthesis}\label{sec:synth}
In this section, it is shown how to set the prediction horizon $T\in\mathbb{N}$, and construct a sequence of terminal penalty matrices, to achieve a performance loss specification for the receding-horizon policy~\eqref{eq:MPC_policy}.
The approach is based on a result, taken from~\cite{sun2023riccati}, that establishes a strict contraction property of the Riccati operator given in~\eqref{eq:ricc_def}.

Given $T \in \mathbb{N}$, the proposed terminal penalty matrix sequence  $(X_{t+T})_{t\in\mathbb{N}_0}$ is given by 
\begin{align}\label{eq:terminal_candidate}
X_{t+T} &= \tilde{Q}_{t+T} + \tilde{A}_{t+T}^{\trs}(\tilde{B}_{t+T} \tilde{B}_{t+T}^{\trs})^{-1} \tilde{B}_{t+T} \tilde{R}_{t+T} \nonumber\\
& \qquad\qquad\qquad \times \tilde{B}_{t+T}^{\trs}(\tilde{B}_{t+T} \tilde{B}_{t+T}^{\trs})^{-1} \tilde{A}_{t+T} \in \mathbb{S}^n_{++}.
\end{align}

\begin{remark} \label{rem:X_tT_finite}
In view of Assumptions~\ref{asm:bounded_data} and \ref{asm:ctr_obs_uniform}, with $X_{t+T}$ as per \eqref{eq:terminal_candidate}, $\sup_{t \in \mathbb{N}_0} \|X_{t+T}\|_2 < +\infty$.
\end{remark}

\begin{lemma}\label{lemma:x_hat_combined}
Given $T\in\mathbb{N}$, and $(X_{t+T})_{t\in\mathbb{N}_0}$ as per \eqref{eq:terminal_candidate},  let 
\begin{align} \label{eq:hatX}
\hat{X}_{t+1} := \mathcal{R}_{t+1}\circ \cdots \circ \mathcal{R}_{t+T-1}(X_{t+T}),
\end{align}
for $t\in\mathbb{N}_0$.
Then, 
$\hat{X}_{t} \succeq \mathcal{R}_t(\hat{X}_{t+1})$, and $\hat{X}_{t} \succeq P_{t}$, 
where $(P_{t})_{t\in\mathbb{N}_0}$ is 
the bounded positive definite solution of~\eqref{eq:recursion}. Further, $\sup_{t\in\mathbb{N}_0}\|\hat{X}_{t+1}\|_2 <+\infty$.
\end{lemma}

\begin{proof}
Given~\eqref{eq:terminal_candidate}, it follows from~\eqref{eq:RiccBounded} that $\mathcal{R}_{t+T-1}(X_{t+T}) \preceq X_{t+T-1}$ for all $t\in\mathbb{N}$, and thus, $\hat{X}_t \succeq \mathcal{R}_t(\hat{X}_{t+1})$ by Lemma~\ref{lemma:ric_monotone}. Similarly, $X_{t+T-1} \succeq \mathcal{R}_{t+T-1}(P_{t+T}) = P_{t+T-1} $ for all $t \in \mathbb{N}$, and thus, $\hat{X}_t =\mathcal{R}_{t}\circ \cdots \circ \mathcal{R}_{t+T-2}(X_{t+T-1})
 \succeq \mathcal{R}_{t} \circ \cdots \circ \mathcal{R}_{t+T-2}(P_{t+T-1})=P_t$. Boundedness of $(\hat{X}_{t+1})_{t\in\mathbb{N}_0}$
 also follows from~\eqref{eq:RiccBounded} and Assumptions~\ref{asm:bounded_data} and~\ref{asm:ctr_obs_uniform}.
\end{proof}

The following is taken from~\cite[Thm.~1]{sun2023riccati}. It characterizes a strict contraction property of the Riccati operator~\eqref{eq:ricc_def} with respect to the Reimannian metric in Definition~\ref{def:Riemannian}.

\begin{lemma}\label{lemma:contraction}
For all $t\in\mathbb{N}$, and $Y,Z \in \mathbb{S}_{++}^{n}$,
\begin{align}
\delta(\mathcal{R}_t(Y), \mathcal{R}_t(Z)) \leq \rho_t \cdot \delta(Y,Z),
\end{align}
where $\rho_t = \zeta_t/(\zeta_t + \epsilon_t)< 1$,
\begin{subequations}
\label{eq:zeteps}
\begin{align}
\zeta_t &= \|(\tilde{Q}_t + \tilde{Q}_t\tilde{A}_t^{-1}\tilde{B}_t\tilde{R}_t^{-1}\tilde{B}_t^{\trs}(\tilde{A}_t^{\trs})^{-1}\tilde{Q}_t)^{-1}\|_2 , \label{eq:zeta_t}\\
\epsilon_t &= \lambda_{\min} (\tilde{A}_t^{-1}\tilde{B}_t (\tilde{R}_t\!+\!\tilde{B}_t^{\trs}(\tilde{A}_t^{\trs})\!^{-1}\tilde{Q}_t\tilde{A}_t^{-1}\tilde{B}_t)\!^{-1}\! \tilde{B}_t^{\trs}(\tilde{A}_t^{\trs})\!^{-1} ) . \label{eq:eps_t}
\end{align}
\end{subequations}
\end{lemma}

\begin{lemma}\label{lemma:delta_ub}
Given $T \in \mathbb{N}$, with $(X_{t+T})_{t\in\mathbb{N}_0}$
as per~\eqref{eq:terminal_candidate}, and $(P_{t})_{t\in\mathbb{N}_0}$ 
as the bounded positive definite solution of~\eqref{eq:recursion}, the Riemannian distance $\delta(X_{t+T}, P_{t+T}) \leq \overline{\delta}$ for all $t\in\mathbb{N}_0$, where
\begin{align}\label{eq:delta_ub}
\overline{\delta} &:= \sqrt{n} \log \Big(\sup_{t \in \mathbb{N}_0} \frac{\|X_{t+T}\|_2}{\lambda_{\min}(P_{t+T})} \Big)
\\
&\leq \sqrt{n} \log \Big( \sup_{t\in\mathbb{N}_0} \frac{\|X_{t+T}\|_2}{\lambda_{\min}(\tilde{Q}_{t+T})}\Big) < +\infty.
\label{eq:delta_ub_finite}
\end{align}
\end{lemma}

\begin{proof}
From 
\eqref{eq:RiccBounded}, 
$P_{t+T} = \tilde{\mathcal{R}}_{t+T}(P_{t+T+1}) \preceq X_{t+T}$ for all $t \in \mathbb{N}_0$.
So using Lemmas~\ref{lemma:delta_invariant} and~\ref{lemma:log_norm}, it follows that
\begin{align*}
\delta(X_{t+T},P_{t+T})
&= \|\log(P_{t+T}^{-\frac{1}{2}} X_{t+T} P_{t+T}^{-\frac{1}{2}})\|_F \\
&\leq \sqrt{n} \|\log(P_{t+T}^{-\frac{1}{2}} X_{t+T} P_{t+T}^{-\frac{1}{2}})\|_2 \\
&= \sqrt{n} \log(\|P_{t+T}^{-\frac{1}{2}} X_{t+T} P_{t+T}^{-\frac{1}{2}}\|_2) \\
&\leq \sqrt{n} \log(\|X_{t+T}\|_2 / \lambda_{\min}(P_{t+T})),
\end{align*}
which gives~\eqref{eq:delta_ub}.
Since $P_{t+T} \succeq \tilde{Q}_{t+T}$ for all $t \in \mathbb{N}_0$ (see Remark~\ref{rem:PDricc}), $\|P_{t+T}^{-1}\|_2 \leq 1 / \lambda_{\min}(\tilde{Q}_{t+T})$, which with~\eqref{eq:delta_ub} gives~\eqref{eq:delta_ub_finite}. Finiteness of the bound follows from Remark~\ref{rem:X_tT_finite} and Assumption~\ref{asm:ctr_obs_uniform}.
\end{proof}

The following result builds upon Theorem~\ref{theorem:regret_bound}. For given prediction horizon $T\in\mathbb{N}$, and with the terminal penalty matrix sequence set according to~\eqref{eq:terminal_candidate}, the closed-loop is exponentially stable under the 
$T$-step receding-horizon policy~\eqref{eq:MPC_policy}, with 
bounded performance loss.

\begin{theorem}\label{theorem:regret_T}
Given $T\in\mathbb{N}$, and $(X_{t+T})_{t\in\mathbb{N}_0}$ as per~\eqref{eq:terminal_candidate}, consider the receding-horizon state-feedback control policy~\eqref{eq:MPC_policy}.
For all initial states $\xi \in \mathbb{R}^n$, the 
performance loss as defined in~\eqref{eq:regret_def} satisfies
\begin{align}\label{eq:regret_horizon}
\beta(\xi)
\leq \frac{\overline{\lambda}}{\underline{\lambda}}
\frac{\overline{\omega}}{\underline{\omega}} \overline{\omega}\Big(\exp\Big(\Big(\frac{\overline{\zeta}}{\overline{\zeta}+\underline{\epsilon}}\Big)^{T-1} \overline{\delta}\Big)-1\Big) |\xi|^2
\end{align}
with
\begin{align}
&\overline{\zeta} \!:=\! \sup_{t \in \mathbb{N}_0} \zeta_t
\!<\! + \infty, 
\qquad 
\underline{\epsilon} \!:=\! \inf_{t \in \mathbb{N}_0} \epsilon_t
\! > \! 0, \label{eq:inf_eps}
\end{align}
and $\overline{\omega}, \underline{\omega}, \underline{\lambda}, \overline{\lambda}, \zeta_t,\epsilon_t,\overline{\delta} \in \mathbb{R}_{>0}$ as per 
\eqref{eq:w_ub}, \eqref{eq:w_lb}, \eqref{eq:gamma_def}, \eqref{eq:p_ub}, \eqref{eq:zeteps}, and \eqref{eq:delta_ub}, using $(\hat{X}_{t+1})_{t\in\mathbb{N}_0}$ as per~\eqref{eq:hatX}.
\end{theorem}

\begin{proof}
As observed in the preamble to Theorem~\ref{theorem:mpc_1step_stable}, with $(\hat{X}_{t+1})_{t\in\mathbb{N}_0}$ as per~\eqref{eq:hatX},
the $T$-step receding-horizon control policy~\eqref{eq:MPC_policy} is equivalent to the $1$-step policy~\eqref{eq:onestepmpc}.
Now, by Lemma~\ref{lemma:x_hat_combined},
$\hat{X}_{t} \succeq P_{t}$, and  $\hat{X}_{t} \succeq \mathcal{R}_{t}(\hat{X}_{t+1})$,
for all $t\in\mathbb{N}$.
Further,  
\begin{align*}
&\!\!\delta(\hat{X}_{t},P_{t}) \nonumber\\
&\!\!= \delta(\mathcal{R}_{t} \circ \cdots \circ \mathcal{R}_{t+T-2}(X_{t+T-1}), \mathcal{R}_{t} \circ \cdots \circ \mathcal{R}_{t+T-2}(P_{t+T-1})) \nonumber\\
&\!\!\leq \Big(\frac{\overline{\zeta}}{\overline{\zeta}+\underline{\epsilon}}\Big)^{T-1} \delta(X_{t+T-1}, P_{t+T-1})\leq \Big(\frac{\overline{\zeta}}{\overline{\zeta}+\underline{\epsilon}}\Big)^{T-1} \overline{\delta},
\end{align*}
by repeated application of  Lemma~\ref{lemma:contraction}.
In particular,
Theorem~\ref{theorem:regret_bound} applies, which yields the bound \eqref{eq:regret_horizon}. 

It remains to show that $\overline{\zeta}<+\infty$ and $\underline{\epsilon}>0$. First, observe that for all $Z\in\mathbb{S}_{++}^q$, and $Y\in\mathbb{R}^{p\times q}$ with full row rank,
\begin{align*}
\lambda_{\min}(YZY^\prime) 
&= \min_{|x|=1} x^\prime Y Z Y^\prime x \\
&= \min_{|x|=1} \frac{(Y^\prime x)^\prime Z (Y^\prime x)} {|Y^\prime x|^2} \frac{x^\prime YY^\prime x}{|x|^2} \!\geq\! \lambda_{\min}(Z) \lambda_{\min} (YY^\prime),
\end{align*}
and 
$\lambda_{\max}(YZY^\prime) \!=\!
\max_{|x|=1} x^\prime YZY^\prime x
\!\leq\!  \lambda_{\max}(Z) \lambda_{\max}(YY^\prime)$.
Further, $\lambda_{\max}(Z^{-1})=1\big/\lambda_{\min}(Z)$, and when $p=q$,  $\lambda_{\min}(Z) \leq \lambda_{\min}(Z + YY^\prime)$ and $\lambda_{\max}(Z + YY^\prime) \leq \lambda_{\max}(Z) + \lambda_{\max}(YY^\prime)$. Using these inequalities, it follows from \eqref{eq:inf_eps}, and \eqref{eq:zeteps}, that $\overline{\zeta} = \sup_{t\in\mathbb{N}_0} 1 \big/\lambda_{\min}(\tilde{Q}_t + \tilde{Q}_t\tilde{A}_t^{-1}\tilde{B}_t\tilde{R}_t^{-1}\tilde{B}_t^{\trs}(\tilde{A}_t^{\trs})^{-1}\tilde{Q}_t) \leq 1\big/ \inf_{t\in\mathbb{N}_0} \lambda_{\min}(\tilde{Q}_t)$, and 
\begin{align*}
\underline{\epsilon} &\geq  \inf_{t\in\mathbb{N}_0} \lambda_{\min}((\tilde{R}_t\!+\!\tilde{B}_t^{\trs}(\tilde{A}_t^{\trs})\!^{-1}\tilde{Q}_t\tilde{A}_t^{-1}\tilde{B}_t)\!^{-1})\\
& \qquad\qquad\qquad\qquad \times 
 \lambda_{\min}(\tilde{B}_t\tilde{B}_t^\prime) 
 \lambda_{\min}(\tilde{A}_t^{-1} (\tilde{A}_t^\prime)^{-1}) 
\\
&=
\inf_{t\in\mathbb{N}_0}
\lambda_{\min}(\tilde{B}_t\tilde{B}_t^\prime) \!\!\Big/\!\! \big(\lambda_{\max}(\tilde{R}_t\!+\!\tilde{B}_t^{\trs}(\tilde{A}_t^{\trs})\!^{-1}\tilde{Q}_t\tilde{A}_t^{-1}\tilde{B}_t) \lambda_{\max}(\tilde{A}_t^\prime \tilde{A}_t) \big)\\
&\geq 
\inf_{t\in\mathbb{N}_0}
\frac{\lambda_{\min}(\tilde{B}_t\tilde{B}_t^\prime)}{\lambda_{\max}(\tilde{A}_t^\prime \tilde{A}_t)} 
1\big/\big(\lambda_{\max}(\tilde{R}_t) + \lambda_{\max}(\tilde{Q}_t)
\frac{\lambda_{\max}(\tilde{B}_t^\prime \tilde{B}_t)}{\lambda_{\min}(\tilde{A}_t \tilde{A}_t^\prime)}\big).
\end{align*}
As such, with Assumptions \ref{asm:bounded_data} and \ref{asm:ctr_obs_uniform}, $\overline{\zeta}<+\infty$ and $\underline{\epsilon}>0$.
\end{proof}

\begin{remark}
In Theorem~\ref{theorem:regret_T}, $\overline{\omega}$, $\overline{\lambda}$, $\overline{\delta}$, and $\overline{\zeta}$ can be replaced by any corresponding upper bounds for these quantities. Similarly, $\underline{\omega}$, $\underline{\lambda}$, and $\underline{\epsilon}$ can be replaced by any corresponding lower bounds.
\end{remark}


Theorem~\ref{theorem:regret_T} can be used to set the prediction horizon $T \in \mathbb{N}$
to achieve 
specified infinite-horizon performance degradation with the terminal penalty sequence given by~\eqref{eq:terminal_candidate}.


\begin{corollary}\label{corollary:bound_horizon}
Given $T\in\mathbb{N}$, and $(X_{t+T})_{t\in\mathbb{N}_0}$ as per~\eqref{eq:terminal_candidate}, consider the receding-horizon state-feedback control policy~\eqref{eq:MPC_policy}.
Given
$\overline{\beta} \in \mathbb{R}_{>0}$,   
 the 
 performance loss as defined in~\eqref{eq:regret_def} satisfies $\beta(\xi) \leq \overline{\beta} |\xi|^2$ for all initial states $\xi \in \mathbb{R}^n$, if
\begin{align*}
T \geq \left(\log\left(\log\Big(\frac{\overline{\beta}\cdot \underline{\lambda} \cdot\underline{\omega}}{\overline{\omega}\cdot\overline{\lambda}\cdot\overline{\omega}}  + 1\Big) ^{1/\overline{\delta}}\right)\Big/\log\Big(\frac{\overline{\zeta}}{\overline{\zeta}+\underline{\epsilon}}\Big)\right) + 1 
\end{align*}
with $\overline{\omega}, \underline{\omega}, \underline{\lambda}, \overline{\lambda}, \overline{\delta}, \overline{\zeta}, \underline{\epsilon} \in \mathbb{R}_{>0}$ as per \eqref{eq:w_ub}, 
\eqref{eq:w_lb}, \eqref{eq:gamma_def}, \eqref{eq:p_ub}, \eqref{eq:delta_ub}, 
and \eqref{eq:inf_eps},
using $(\hat{X}_{t+1})_{t\in\mathbb{N}_0}$ as per~\eqref{eq:hatX}.
\end{corollary}

\begin{proof}
With reference to~\eqref{eq:regret_horizon}, manipulating
\begin{align*}
\frac{\overline{\lambda}}{\underline{\lambda}}
\frac{\overline{\omega}}{\underline{\omega}} \overline{\omega}\Big(\exp\Big(\Big(\frac{\overline{\zeta}}{\overline{\zeta}+\underline{\epsilon}}\Big)^{T-1} \overline{\delta}\Big)-1\Big) |\xi|^2 \leq \overline{\beta}|\xi|^2
\end{align*}
yields the result.
\end{proof}

\section{Conclusion}\label{sec:conc}
A link is established between prediction horizon length and the infinite-horizon performance loss of a receding-horizon approximation of the optimal policy for a lifted reformulation of \eqref{eq:opt_problem}. This is achieved via a strict contraction property of the corresponding Riccati operator, under an assumed uniform controllability and uniform observability property of the dynamics and stage cost in the original domain. Ongoing work is focused on extending the approach to accommodate cross-terms in the stage cost, uncertainty in the problem data, and constraints on the input and state variables.

\printbibliography

@article{grune2008infinite,
  title={On the infinite horizon performance of receding horizon controllers},
  author={Grune, Lars and Rantzer, Anders},
  journal={IEEE Transactions on Automatic Control},
  volume={53},
  number={9},
  pages={2100--2111},
  year={2008},
  publisher={IEEE}
}

@book{bertsekas2012dynamic,
  title={{D}ynamic {P}rogramming and {O}ptimal {C}ontrol: {V}olume I},
  author={Bertsekas, Dimitri},
  volume={1},
  year={2012},
  publisher={Athena scientific}
}

@book{anderson2007optimal,
  title={{O}ptimal {C}ontrol: {L}inear {Q}uadratic {M}ethods},
  author={Anderson, Brian DO and Moore, John B},
  year={2007},
  publisher={Courier Corporation}
}

@inproceedings{zhang2021regret,
  title={On the regret analysis of online {LQR} control with predictions},
  author={Zhang, Runyu and Li, Yingying and Li, Na},
  booktitle={2021 American Control Conference (ACC)},
  pages={697--703},
  year={2021},
  organization={IEEE}
}

@article{de1992time,
  title={On the time-varying {R}iccati difference equation of optimal filtering},
  author={De Nicolao, Giuseppe},
  journal={SIAM Journal on Control and Optimization},
  volume={30},
  number={6},
  pages={1251--1269},
  year={1992},
  publisher={SIAM}
}

@article{bougerol1993kalman,
  title={Kalman filtering with random coefficients and contractions},
  author={Bougerol, Philippe},
  journal={SIAM Journal on Control and Optimization},
  volume={31},
  number={4},
  pages={942--959},
  year={1993},
  publisher={SIAM}
}

@article{li2022performance,
  title={Performance Bounds of Model Predictive Control for Unconstrained and Constrained Linear Quadratic Problems and Beyond},
  author={Li, Yuchao and Karapetyan, Aren and Lygeros, John and Johansson, Karl H and M{\aa}rtensson, Jonas},
  journal={arXiv preprint arXiv:2211.06187},
  year={2022}
}

@article{lin2021perturbation,
  title={Perturbation-based regret analysis of predictive control in linear time varying systems},
  author={Lin, Yiheng and Hu, Yang and Shi, Guanya and Sun, Haoyuan and Qu, Guannan and Wierman, Adam},
  journal={Advances in Neural Information Processing Systems},
  volume={34},
  pages={5174--5185},
  year={2021}
}

@incollection{bitmead1991riccati,
  title={Riccati difference and differential equations: Convergence, monotonicity and stability},
  author={Bitmead, Robert R and Gevers, Michel},
  booktitle={The Riccati Equation},
  pages={263--291},
  year={1991},
  publisher={Springer}
}

@article{keerthi1988optimal,
  title={Optimal infinite-horizon feedback laws for a general class of constrained discrete-time systems: Stability and moving-horizon approximations},
  author={Keerthi, S Sathiya and Gilbert, Elmer Grant},
  year={1988},
  journal={Journal of Optimization Theory and Applications}, 
  pages={265--293},
  volume={57}
}

@article{bof2018lyapunov,
  title={Lyapunov theory for discrete time systems},
  author={Bof, Nicoletta and Carli, Ruggero and Schenato, Luca},
  journal={arXiv preprint arXiv:1809.05289},
  year={2018}
}

@article{sun2023riccati,
  title={On {R}iccati contraction in time-varying linear-quadratic control},
  author={Sun, Jintao and Cantoni, Michael},
  journal={arXiv preprint arXiv:2305.06003},
  year={2023}
}

\appendix

\begin{proof}[Proof of lemma~\ref{lemma:logm}]
Note that
\begin{align*}
& \exp(T^{-1}\diag(\log(\lambda_1),\dots,\log(\lambda_n))T) \\
&= \sum_{k\in\mathbb{N}} \frac{1}{k!} (T^{-1}\diag(\log(\lambda_1),\dots,\log(\lambda_n))T)^k \\
&= T^{-1}\diag(\exp(\log(\lambda_1)),\dots,\exp({\log(\lambda_n)}))T = M. \qedhere
\end{align*}
\end{proof}

\begin{proof}[Proof of Lemma~\ref{lemma:log_symmetric}]
Since $M\in\mathbb{S}_{++}^n$, there exists orthogonal $T^\prime=T^{-1}$ such that $M=T^\prime\diag(\lambda_1, \dots, \lambda_n) T$ with~$\{\lambda_1, \ldots, \lambda_n\} \subset\mathbb{R}_{>0}$. From Lemma~\ref{lemma:logm}, $\log(M)=T^\prime\diag(\log\lambda_1, \dots, \log\lambda_n)T$, which is symmetric.
\end{proof}

\begin{proof}[Proof of Lemma~\ref{lemma:delta_invariant}]
Note that since $Y, Z^{-1} \in \mathbb{S}^n_{++}$, the spectrum of both $YZ^{-1}\in\mathbb{R}^{n\times n}$, and $Z^{-\frac{1}{2}}YZ^{-\frac{1}{2}}\in\mathbb{S}_{++}^n$, excludes $0$, whereby
$\mathrm{spec}\{YZ^{-1}\} = \mathrm{spec}\{Z^{-\frac{1}{2}}YZ^{-\frac{1}{2}}\} = \{\lambda_1, \ldots, \lambda_n\} \subset \mathbb{R}_{>0}$.
Further, there exists orthogonal matrix $T^{\trs} = T^{-1}$ such that
$Z^{-\frac{1}{2}}YZ^{-\frac{1}{2}} = T^{\trs} \Lambda T$,
where
$\Lambda = \diag(\lambda_1, \dots, \lambda_n)$.
It follows that
\begin{align*}
\delta(Y,Z) = \sqrt{\sum_{i=1}^{n} \log^{2} \lambda_i} 
&= \sqrt{\tr((\log(\Lambda))^2)} \\
&= \sqrt{\tr(T^{\trs} \log(\Lambda) T T^{\trs} \log(\Lambda) T)} \\
&= \sqrt{\tr((\log(Z^{-\frac{1}{2}}YZ^{-\frac{1}{2}}))^2)} .
\end{align*}
From Lemma~\ref{lemma:log_symmetric}, $\log(Z^{-\frac{1}{2}}YZ^{-\frac{1}{2}}) \in \mathbb{S}^n$. Thus,
\begin{align*}
\delta(Y,Z) &= \sqrt{\tr((\log(Z^{-\frac{1}{2}}YZ^{-\frac{1}{2}}))^{\trs} \log(Z^{-\frac{1}{2}}YZ^{-\frac{1}{2}}))} \\
&= \|\log(Z^{-\frac{1}{2}}YZ^{-\frac{1}{2}})\|_F . \qedhere
\end{align*}
\end{proof}

\begin{proof}[Proof of Lemma~\ref{lemma:log_norm}]
Since $M\in\mathbb{S}^n$, there exists orthogonal matrix $T^{\trs} = T^{-1}$ such that
$M = T^{\trs} \diag(\lambda_1,\dots,\lambda_n) T$,
where $\{\lambda_1,\dots,\lambda_n\} \subset \mathbb{R}_{\geq 1}$ is the spectrum of $M$. It follows from Lemma~\ref{lemma:logm} that
$\log(M) = T^{\trs} \diag(\log(\lambda_1),\dots,\log(\lambda_n)) T$,
and
\begin{align*}
\|\log(M)\|_2 = \max_i (\log(\lambda_i))
&= \log(\max_i (\lambda_i)) = \log(\|M\|_2),
\end{align*}
where the second equality holds because $\lambda_i\geq 1$.
\end{proof}

\begin{proof}[Proof of Lemma~\ref{lemma:u-v_bound}]
Since $Y-Z \in \mathbb{S}^n_{++}$,
\begin{align*}
\|Y-Z\|_2
&= \|Z^{\frac{1}{2}}(Z^{-\frac{1}{2}}YZ^{-\frac{1}{2}} - I_n)Z^{\frac{1}{2}}\|_2 \\
&\leq \|Z\|_2 \|Z^{-\frac{1}{2}}YZ^{-\frac{1}{2}} - I_n\|_2 \\
&= \|Z\|_2 (\|Z^{-\frac{1}{2}}YZ^{-\frac{1}{2}}\|_2 - 1) \\
&= \|Z\|_2 (\exp({\log(\|Z^{-\frac{1}{2}}YZ^{-\frac{1}{2}}\|_2})) - 1) .
\end{align*}
Note that
\begin{align*}
\lambda_{\min}(Z^{-\frac{1}{2}}YZ^{-\frac{1}{2}})-1
&= \lambda_{\min}(Z^{-\frac{1}{2}}YZ^{-\frac{1}{2}}-I_n) \\
&= \lambda_{\min}(Z^{-\frac{1}{2}}(Y-Z)Z^{-\frac{1}{2}}) \\
& \geq 0,
\end{align*}
which with Lemma~\ref{lemma:log_norm} gives
\begin{align*}
\log(\|Z^{-\frac{1}{2}}YZ^{-\frac{1}{2}}\|_2) = \|\log(Z^{-\frac{1}{2}}YZ^{-\frac{1}{2}})\|_2 .
\end{align*}
Since $\|\cdot\|_2 \leq \|\cdot\|_F$, it follows that
\begin{align*}
\|Y - Z\|_2 &\leq \|Z\|_2 (\exp(\|\log(Z^{-\frac{1}{2}}YZ^{-\frac{1}{2}})\|_2) - 1) \\
&\leq \|Z\|_2 (\exp(\|\log(Z^{-\frac{1}{2}}YZ^{-\frac{1}{2}})\|_F) - 1) ,
\end{align*}
which combined with~\eqref{eq:delta_equality} gives the claimed result.
\end{proof}

\end{document}